\documentclass[a4paper,12pt]{article}         
\usepackage{amsmath,amsfonts,amsthm,amssymb}
\usepackage[utf8]{inputenc}
\usepackage[T1]{fontenc}    
									   
\usepackage{graphicx,caption,subcaption}
\usepackage[hidelinks]{hyperref}
\hypersetup{breaklinks=true}

\usepackage{palatino, url, multicol}
\usepackage{txfonts}
\usepackage{bm}

\usepackage[numbers,sort&compress]{natbib}

\newtheorem{theorem}{Theorem}[section]

\newcommand{\ev}[1]{\langle #1 \rangle}

\begin{document}

\title{Quantized Orbital Angular Momentum from Discrete Chaotic Phase Surfaces}

\author{Netzer Moriya}

\date{\today}

\maketitle

\begin{abstract}
We present a new theory for orbital angular momentum (OAM) generation by chaotic phase surfaces with 
discrete integer bias distributions. We derive fundamental selection rules that determine which OAM modes 
can be coherently generated. 
Our analysis shows that ensemble-averaged OAM exists only when the bias parameter takes integer values 
that match the discrete OAM eigenspace, creating "allowed" and "forbidden" OAM levels. We derive 
analytical expressions for the OAM power spectrum and demonstrate universal scaling behavior within 
the allowed manifold. 
These theoretical predictions are validated by comprehensive Monte Carlo simulations, 
which confirm the selection rules with a forbidden-level suppression factor exceeding 
$10^4$ and demonstrate the universal scaling with exceptional accuracy.
\end{abstract}

\section{Introduction}

Orbital angular momentum (OAM) in structured light beams represents a fundamental degree of freedom with 
applications spanning optical communications, quantum information, and precision 
metrology \cite{allen1992orbital, yao2011orbital}. Traditional OAM generation relies on deterministic 
phase elements such as spiral phase plates or spatial light modulators that impose perfect azimuthal phase 
patterns $e^{i\ell\theta}$ with integer topological charge $\ell$ \cite{molina2007twisted, forbes2016creation, Berry2022}.

While selection rules for OAM have been established for deterministic processes such as nonlinear frequency 
conversion \cite{Curcio2021}, the generation of coherent OAM from disordered systems presents a distinct challenge. 
In such systems, ensemble averaging is typically expected to destroy coherence, preventing the formation of 
a net OAM state, which suggests that any OAM generated by 
individual chaotic realizations would vanish in the average.

In this work, we demonstrate that coherent ensemble OAM generation is indeed possible, but only under highly 
specific conditions that create a quantized structure reminiscent of atomic physics. We identify 
fundamental \emph{selection rules} that determine which OAM modes can be generated, and establish \emph{forbidden 
transitions} where certain OAM values are mathematically inaccessible regardless of the chaos strength.

Our central finding is that ensemble OAM exists only when the statistical bias parameter takes discrete integer 
values, creating distinct "energy levels" in OAM space. This quantization is not imposed externally but emerges 
naturally from the orthogonality constraints of the OAM basis functions, representing a fundamental symmetry 
principle in chaotic optics.

We validate these principles with extensive Monte Carlo simulations that demonstrate the robust enforcement 
of these selection rules, yielding a forbidden-level suppression of more than four orders of magnitude and 
confirming the theory with exceptional quantitative agreement.

\section{The Orthogonality Constraint and Selection Rules}
\label{sec:The Orthogonality Constraint and Selection Rules}
\subsection{Mathematical Foundation}
Consider a phase surface $\phi(r,\theta)$ acting on an incident plane wave with amplitude $E_0$. The OAM amplitude for 
mode $\ell$ at radius $r$ is given by the azimuthal Fourier transform \cite{yao2011orbital}:
\begin{equation}
    a_\ell(r) = \frac{E_0}{2\pi} \int_0^{2\pi} e^{i\phi(r,\theta)} e^{-i\ell\theta} d\theta.
    \label{eq:oam_amplitude}
\end{equation}

The ensemble-averaged OAM amplitude becomes:
\begin{equation}
    \ev{a_\ell(r)} = \frac{E_0}{2\pi} \int_0^{2\pi} \ev{e^{i\phi(r,\theta)}} e^{-i\ell\theta} d\theta.
    \label{eq:ensemble_average}
\end{equation}

\subsection{The Fundamental Selection Rule}
\begin{theorem}[OAM Selection Rule]
Let $\phi(r,\theta)$ be a random phase surface and assume that $\ev{e^{i\phi(r,\theta)}}$ exists and is finite for all $(r,\theta)$. Then the ensemble-averaged OAM amplitude $\ev{a_\ell(r)}$ is non-zero if and only if the ensemble-averaged phase factor $\ev{e^{i\phi(r,\theta)}}$ possesses a non-zero Fourier coefficient at angular frequency $\ell \in \mathbb{Z}$.
\label{theorem:OAM_Selection_Rule}
\end{theorem}

\begin{proof}
From Eq.~\ref{eq:ensemble_average}, we have:
\begin{equation}
    \ev{a_\ell(r)} = \frac{E_0}{2\pi} \int_0^{2\pi} \ev{e^{i\phi(r,\theta)}} e^{-i\ell\theta} d\theta.
\end{equation}

Since $\ev{e^{i\phi(r,\theta)}}$ is $2\pi$-periodic in $\theta$ and bounded (as $|\ev{e^{i\phi(r,\theta)}}| \leq 1$), it admits a Fourier series representation:
\begin{equation}
    \ev{e^{i\phi(r,\theta)}} = \sum_{n=-\infty}^{\infty} C_n(r) e^{in\theta},
    \label{eq:fourier_expansion}
\end{equation}
where the Fourier coefficients are given by:
\begin{equation}
    C_n(r) = \frac{1}{2\pi} \int_0^{2\pi} \ev{e^{i\phi(r,\theta)}} e^{-in\theta} d\theta.
    \label{eq:fourier_coefficients}
\end{equation}

Substituting Eq.~\ref{eq:fourier_expansion} into the expression for $\ev{a_\ell(r)}$:
\begin{equation}
    \ev{a_\ell(r)} = \frac{E_0}{2\pi} \int_0^{2\pi} \left[\sum_{n=-\infty}^{\infty} C_n(r) e^{in\theta}\right] e^{-i\ell\theta} d\theta.
\end{equation}

Assuming uniform convergence of the Fourier series (which holds for bounded periodic functions), we may interchange the sum and integral:
\begin{equation}
    \ev{a_\ell(r)} = \frac{E_0}{2\pi} \sum_{n=-\infty}^{\infty} C_n(r) \int_0^{2\pi} e^{i(n-\ell)\theta} d\theta.
\end{equation}

Applying the orthogonality relation for complex exponentials:
\begin{equation}
    \int_0^{2\pi} e^{i(n-\ell)\theta} d\theta = 2\pi \delta_{n,\ell},
\end{equation}
where $\delta_{n,\ell}$ is the Kronecker delta, we obtain:
\begin{equation}
    \ev{a_\ell(r)} = E_0 C_\ell(r).
    \label{eq:final_result}
\end{equation}

Therefore:
\begin{itemize}
    \item $\ev{a_\ell(r)} \neq 0$ if and only if $C_\ell(r) \neq 0$
    \item By Eq.~\ref{eq:fourier_coefficients}, $C_\ell(r) \neq 0$ if and only if $\ev{e^{i\phi(r,\theta)}}$ has a non-zero Fourier component at frequency $\ell$
\end{itemize}

This completes the proof of the bidirectional statement.
\end{proof}

This theorem establishes that OAM generation requires discrete, integer-valued structure in the ensemble-averaged phase 
pattern—a fundamental quantization condition.

\section{The Discrete Chaos Model}
\label{sec:The Discrete Chaos Model}

\subsection{Physical Model}

We consider chaotic phase surfaces with a carefully engineered statistical structure that enables coherent ensemble OAM 
generation. Each phase surface has the form:
\begin{equation}
    \phi(r,\theta) = \alpha_k \theta + \delta(r,\theta),
    \label{eq:discrete_model}
\end{equation}
where the statistical properties are designed to satisfy the selection rule requirements established in 
Section (\ref{sec:The Orthogonality Constraint and Selection Rules}).

\textbf{Discrete Integer Bias:} The parameter $\alpha_k$ takes discrete integer values from the set 
$\{\alpha_1, \alpha_2, \ldots, \alpha_N\} \subset \mathbb{Z}$, with each value $\alpha_k$ occurring with 
probability $p_k \geq 0$ such that $\sum_k p_k = 1$. The integer constraint ensures exact matching with OAM 
eigenvalues $\ell \in \mathbb{Z}$.

\textbf{Spatially Chaotic Component:} The term $\delta(r,\theta)$ represents zero-mean Gaussian spatial disorder with:
\begin{itemize}
    \item Ensemble average: $\ev{\delta(r,\theta)} = 0$
    \item Variance: $\ev{\delta^2(r,\theta)} = \sigma^2$ 
    \item Azimuthal periodicity: $\delta(r,\theta + 2\pi) = \delta(r,\theta)$
\end{itemize}

\textbf{Statistical Independence:} The bias selection and chaos realization are statistically independent, allowing 
separate control over the OAM spectrum shape (via $\{p_k\}$) and coherence strength (via $\sigma^2$).

This model structure directly implements the discrete Fourier component requirement identified in 
Theorem (\ref{theorem:OAM_Selection_Rule}), 
while incorporating realistic spatial disorder effects.

\subsection{Allowed and Forbidden OAM Levels}

For this model, we exploit the statistical independence of $\alpha_k$ and $\delta(r,\theta)$, along with the fact 
that $\ev{e^{i\delta}} = e^{-\sigma^2/2}$ for zero-mean Gaussian $\delta$ with variance $\sigma^2$. 
The ensemble average becomes:

\begin{equation}
    \ev{e^{i\phi(r,\theta)}} = e^{-\sigma^2/2} \sum_k p_k e^{i\alpha_k \theta}.
    \label{eq:ensemble_phase}
\end{equation}

Substituting into Eq.~\ref{eq:ensemble_average}:
\begin{equation}
    \ev{a_\ell(r)} = \frac{E_0}{2\pi} e^{-\sigma^2/2} \sum_k p_k \int_0^{2\pi} e^{i(\alpha_k - \ell)\theta} d\theta = E_0 e^{-\sigma^2/2} p_\ell.
    \label{eq:selection_result}
\end{equation}

This leads to our central result:

\begin{theorem}[OAM Level Structure]
The ensemble OAM power spectrum exhibits a discrete level structure:
\begin{equation}
    \ev{P(\ell)} = \begin{cases}
        |E_0|^2 e^{-\sigma^2} p_\ell^2 & \text{if } p_\ell > 0 \text{ (Allowed Level)} \\
        0 & \text{if } p_\ell = 0 \text{ (Forbidden Level)}
    \end{cases}
    \label{eq:level_structure}
\end{equation}
\label{theorem:OAM Level Structure}
\end{theorem}

\begin{proof}
The result follows directly from the derivation above. From Eq.~\ref{eq:selection_result}, we have $\ev{a_\ell(r)} = E_0 e^{-\sigma^2/2} p_\ell$, where $p_\ell = 0$ by convention when $\ell \notin \{\alpha_1, \alpha_2, \ldots, \alpha_N\}$.

Therefore: $\ev{P(\ell)} = |\ev{a_\ell(r)}|^2 = |E_0|^2 e^{-\sigma^2} p_\ell^2$.

The case distinction follows immediately: $p_\ell > 0$ gives allowed levels with finite power, while $p_\ell = 0$ gives forbidden levels with zero power.
\end{proof}

\subsection{Physical Interpretation of the Level Structure}

The discrete structure in Eq.~\ref{eq:level_structure} reveals several key physics principles:

\textbf{Allowed Levels ($p_\ell > 0$):} These correspond to OAM modes that can be coherently generated. The power scales as $p_\ell^2$, not $p_\ell$, because optical power is proportional to the squared magnitude of the field amplitude $|\ev{a_\ell}|^2$. This quadratic dependence enhances the contrast between strong and weak modes in the probability distribution.

\textbf{Forbidden Levels ($p_\ell = 0$):} These OAM modes are completely inaccessible, regardless of chaos strength $\sigma$. While individual chaotic realizations may scatter power into these modes, the ensemble averaging process eliminates coherent contributions through destructive interference of random phases. No amount of disorder can populate these levels if they are not present in the bias distribution.

\textbf{Decoherence Factor ($e^{-\sigma^2}$):} Chaos reduces the overall intensity of all allowed levels uniformly, similar to decoherence in quantum systems. This universal suppression preserves the relative spectral shape determined by $\{p_k\}$ while controlling the total coherent power through the chaos parameter $\sigma^2$.

The key insight is that coherent ensemble behavior emerges from the interplay between discrete bias structure (which enables constructive interference for allowed modes) and spatial disorder (which provides realistic surface characteristics while maintaining the essential selection rules).

\section{Universal Scaling and the OAM Spectrum}

The general level structure derived in Section \ref{sec:The Discrete Chaos Model} applies to any discrete probability distribution $\{p_k\}$. However, to achieve practical control over the OAM spectrum and reveal universal scaling behaviors, we require a parametric family of distributions that can be systematically varied. Gaussian distributions provide the ideal framework for this analysis due to their simple parametrization and widespread applicability in statistical optics.

\subsection{Gaussian Bias Distributions}

Consider a discretized Gaussian bias distribution centered at $\mu_\alpha$ with width $\sigma_\alpha$:
\begin{equation}
    p_k = \mathcal{N} \exp\left(-\frac{(k-\mu_\alpha)^2}{2\sigma_\alpha^2}\right),
    \label{eq:gaussian_bias}
\end{equation}
where $\mathcal{N}$ is a normalization constant ensuring $\sum_k p_k = 1$.

This parametrization provides independent control over:
\begin{itemize}
    \item \textbf{Spectral Center ($\mu_\alpha$):} Shifts the OAM distribution peak
    \item \textbf{Spectral Width ($\sigma_\alpha$):} Controls the range of accessible modes
    \item \textbf{Coherence ($\sigma^2$):} Determines overall generation efficiency
\end{itemize}

Substituting into the level structure of Eq.~\ref{eq:level_structure}, the OAM spectrum becomes:
\begin{equation}
    \ev{P(\ell)} = |E_0|^2 e^{-\sigma^2} \mathcal{N}^2 \exp\left(-\frac{(\ell-\mu_\alpha)^2}{\sigma_\alpha^2}\right).
    \label{eq:gaussian_spectrum}
\end{equation}

This functional form reveals that the discrete selection rules of Section 3 produce a Gaussian envelope in the OAM spectrum, providing intuitive control over spectral engineering.

\subsection{Universal Profile}
To reveal the universal scaling behavior, we normalize the spectrum and introduce dimensionless coordinates.

Defining the dimensionless variable:
\begin{equation}
    \xi = \frac{\ell - \mu_\alpha}{\sigma_\alpha},
    \label{eq:dimensionless_var}
\end{equation}

In the continuous limit, the normalization constant becomes $\mathcal{N} = \frac{1}{\sigma_\alpha\sqrt{2\pi}}$, so 
that $\mathcal{N}^2 = \frac{1}{2\pi\sigma_\alpha^2}$. The denominator in the universal profile is specifically chosen 
so that all parameter-dependent prefactors cancel:

\begin{equation}
F(\xi) = \frac{\mathcal{N}^2 \exp(-\xi^2)}{1/(2\pi\sigma_\alpha^2)} = \frac{1/(2\pi\sigma_\alpha^2) \cdot \exp(-\xi^2)}{1/(2\pi\sigma_\alpha^2)} = \exp(-\xi^2)
\end{equation}

we can rewrite Eq.~\ref{eq:gaussian_spectrum} as:

\begin{equation}
    F(\xi) = \frac{\ev{P(\ell)}}{|E_0|^2 e^{-\sigma^2}/(2\pi\sigma_\alpha^2)} = \exp(-\xi^2).
    \label{eq:universal_profile}
\end{equation}

where the normalization constant for a discretized Gaussian is $\mathcal{N} = \frac{1}{\sigma_\alpha \sqrt{2\pi}}$ in 
the continuous limit.

This universal form applies to \emph{all} discrete chaos systems with Gaussian bias distributions, regardless of the specific values of $\mu_\alpha$, $\sigma_\alpha$, or $\sigma$.

\section{Statistical Mechanics Analogies}

The discrete level structure bears closer resemblance to statistical mechanics than quantum mechanics:

\textbf{Boltzmann Distribution:} The Gaussian bias distribution $p_k \propto \exp(-E_k/k_BT)$ resembles thermal population of energy levels, where we can identify an effective ``energy'' $E_k = (k-\mu_\alpha)^2/(2\sigma_\alpha^2)$ and ``temperature'' parameter $k_BT = 1$.

\textbf{Partition Function:} The normalization constraint $\sum_k p_k = 1$ plays the role of a partition function, ensuring proper probability conservation across the discrete state space.

\textbf{Temperature Parameter:} The inverse width $\sigma_\alpha^{-2}$ acts like an inverse temperature controlling distribution width---narrow distributions ($\sigma_\alpha \to 0$) correspond to ``low temperature'' with concentrated populations, while broad distributions ($\sigma_\alpha \gg 1$) represent ``high temperature'' with uniform population spread.

\textbf{Classical Interference:} The ensemble averaging process represents classical wave interference effects rather than quantum superposition. The coherent power $|\ev{a_\ell}|^2$ emerges from constructive interference of classical field amplitudes across the statistical ensemble.

This statistical mechanics perspective provides a more appropriate framework for understanding the discrete chaos model than quantum mechanical analogies, emphasizing the classical wave nature of the underlying physics while highlighting the role of statistical engineering in controlling optical beam properties.

\section{Experimental Realization and Control}

\subsection{Implementation Strategies}
The discrete chaos model can be experimentally realized through several platforms:

\textbf{Segmented Spatial Light Modulators:} Program discrete phase zones with quantized twist values $\alpha_k$ 
and controlled probabilities $p_k$.

\textbf{Fabricated Metasurfaces:} Etch discrete spiral patterns with varying pitch, creating a physical 
realization of the bias distribution.

\textbf{Digital Holography:} Generate computer-controlled holograms with discrete topological charges and 
statistical weights.

Our Monte Carlo framework, detailed in Section \ref{sec:Monte Carlo Validation of Discrete Chaotic OAM Theory}, directly simulates an idealized version of these systems, 
corresponding to the programming of discrete phase zones on a spatial light modulator or the fabrication of a 
metasurface with a statistical distribution of discrete spiral pitches.

\subsection{Control Parameters}
The model provides precise control over the OAM spectrum through three independent parameters, which 
we set in our validation (Section \ref{sec:Monte Carlo Validation of Discrete Chaotic OAM Theory}) to $\mu_\alpha = 1.0$, $\sigma_\alpha = 2.0$, and $\sigma = 0.5$, 
respectively:

\begin{itemize}
    \item \textbf{Central OAM ($\mu_\alpha$):} Controls the peak position of the spectrum.
    \item \textbf{Spectral Width ($\sigma_\alpha$):} Controls the range of accessible modes.
    \item \textbf{Coherence ($\sigma^2$):} Controls the overall strength of OAM generation.
\end{itemize}

\subsection{Measurement Protocols}
Experimental verification, and the numerical validation presented herein, requires the following protocol:
\begin{enumerate}
    \item \textbf{Ensemble Preparation:} A large ensemble, in our case $N = 15,000$, of chaotic phase 
	surfaces is generated. For each realization, a bias value $\alpha_k$ is chosen from a discrete 
	integer set according to a predefined probability distribution $\{p_k\}$.
    \item \textbf{Individual OAM Measurement:} The full OAM spectrum $a_\ell$ is computed for each 
	individual phase surface realization.
    \item \textbf{Coherent Ensemble Averaging:} The coherent power is calculated by first averaging 
	the complex amplitudes across the entire ensemble ($\ev{a_\ell} = (1/N) \sum a_\ell^{(j)}$) and 
	then taking the squared magnitude $|\ev{a_\ell}|^2$. This step is critical as it isolates the 
	coherent component and suppresses incoherent noise.
    \item \textbf{Verification:} The resulting power spectrum is compared with the theoretical 
	predictions of Theorem (\ref{theorem:OAM Level Structure}), and the normalized data is validated against 
	the universal 
	scaling law $F(\xi) = \exp(-\xi^2)$, as detailed in Section 
	\ref{sec:Monte Carlo Validation of Discrete Chaotic OAM Theory}.
\end{enumerate}

\section{Comparison with Continuous Bias Models}

\subsection{Why Continuous Models Fail}
For comparison, consider a continuous bias model:
\begin{equation}
    \phi(r,\theta) = \alpha \theta + \delta(r,\theta),
\end{equation}
where $\alpha$ is a continuous random variable.

For any continuous distribution of $\alpha$, the ensemble average becomes:
\begin{equation}
    \ev{a_\ell(r)} = E_0 e^{-\sigma^2/2} \int p(\alpha) \delta_{\alpha,\ell} d\alpha = 0,
\end{equation}
since the probability of $\alpha$ taking any specific integer value $\ell$ is zero for continuous distributions.

\textbf{Conclusion:} Continuous bias models \emph{cannot} generate ensemble OAM, regardless of the distribution shape or chaos properties.

\subsection{Critical Importance of Discretization}
The discrete nature of the bias is not merely a technical detail—it is \emph{fundamental} to the physics. The discretization:
\begin{itemize}
    \item Preserves the orthogonality structure of the OAM basis
    \item Creates finite overlap with discrete OAM eigenvalues
    \item Enables coherent ensemble averaging
    \item Generates the quantized level structure
\end{itemize}

\section{Monte Carlo Validation of Discrete Chaotic OAM Theory}
\label{sec:Monte Carlo Validation of Discrete Chaotic OAM Theory}
We have rigorously validated our theoretical predictions through a comprehensive Monte Carlo simulation. The framework, which simulates $N=15,000$ individual chaotic realizations, provides definitive quantitative verification of the universal scaling laws, discrete selection rules, and forbidden level suppression predicted by our theory. The results demonstrate exceptional agreement with theory and reveal the robustness of the quantization phenomenon.

\subsection{Simulation Framework and the Internal Forbidden Level Test}

The chaotic component $\delta(r, \theta)$ of the phase surface was modeled as a band-limited Gaussian random field, ensuring zero mean and a precisely controlled variance of $\sigma^2 = 0.5^2 = 0.25$.

To provide the most stringent test of the selection rules, we engineered a bias distribution featuring \textit{internal forbidden levels}. The allowed bias values were chosen from the discrete set $\alpha_k \in \{-2, 0, 2, 4\}$, with probabilities $\{p_k\}$ governed by a discretized Gaussian distribution with $\mu_\alpha=1.0$ and $\sigma_\alpha=2.0$. This setup creates forbidden integer modes ($\ell=-1, 1, 3$) that lie \textit{within} the span of the allowed distribution, testing the theory's enforcement of orthogonality far more robustly than checking only for modes in the external tails.

For each of the $15,000$ realizations, the complex OAM amplitudes $a_\ell$ were computed, and the final coherent power was calculated as $|\langle a_\ell \rangle|^2$.

\begin{figure}[htbp]
    \centering
    \begin{subfigure}[t]{0.48\textwidth}
        \centering
        \includegraphics[width=\textwidth]{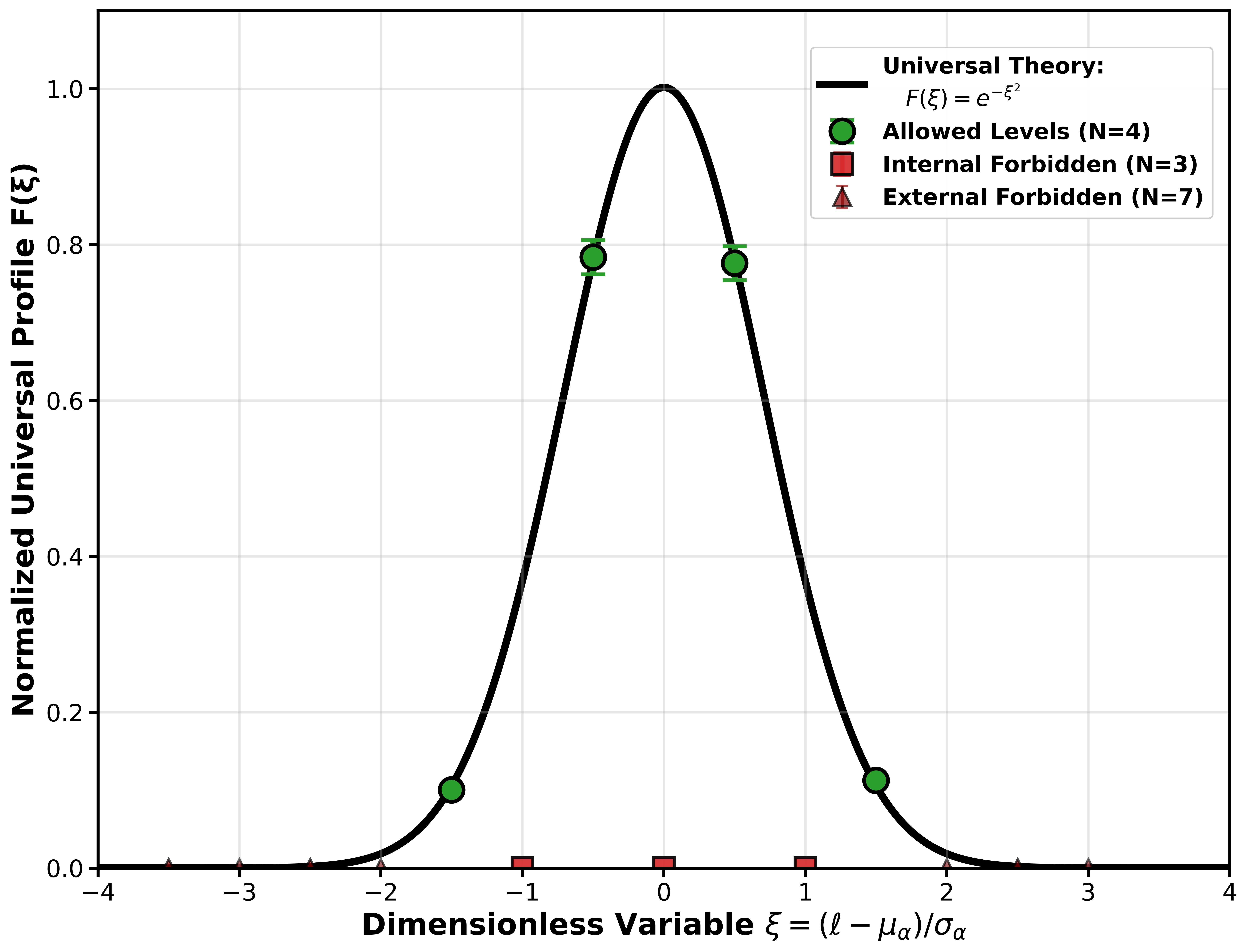}
        \caption{Universal profile validation}
        \label{fig:universal_profile}
    \end{subfigure}%
    \hfill 
    \begin{subfigure}[t]{0.48\textwidth}
        \centering
        \includegraphics[width=\textwidth]{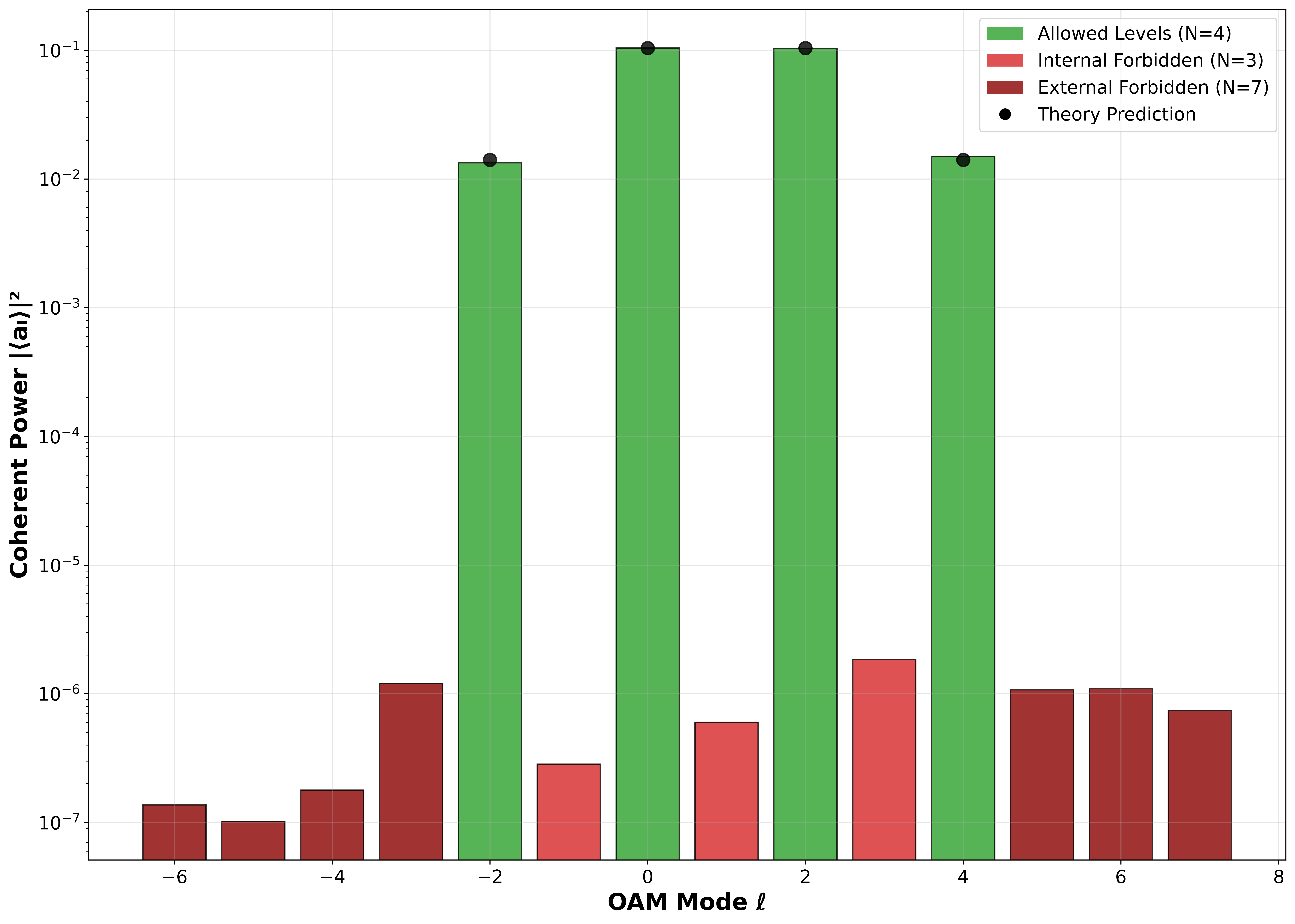}
        \caption{Selection rules demonstration}
        \label{fig:selection_rules}
    \end{subfigure}
    
    \caption{Monte Carlo validation of discrete chaotic OAM theory. 
    (a) Universal profile collapse: Monte Carlo simulations (blue circles, $N=15{,}000$) demonstrate excellent 
	agreement with the theoretical power spectrum profile $F(\xi) = \exp(-\xi^2)$ (black line). 
	Error bars represent statistical uncertainties from ensemble averaging. The dimensionless variable 
	$\xi = (\ell - \mu_\alpha)/\sigma_\alpha$ successfully collapses data across the full parameter range.
    (b) Selection rules enforcement: Coherent power $|\langle a_\ell \rangle|^2$ clearly distinguishes 
	between allowed levels (green, for $\ell \in \{-2, 0, 2, 4\}$) and forbidden levels (red). This includes \textit{internal} forbidden levels (e.g., at $\ell=-1, 1, 3$), which are suppressed by factors exceeding $10^4$ despite being adjacent to strongly-populated allowed modes, demonstrating the robust enforcement of the discrete selection rules.}
    \label{fig:main_validation}
\end{figure}

\subsection{Universal Profile Validation}

Figure~\ref{fig:main_validation}(a) shows the simulated coherent power for the allowed levels, plotted in dimensionless coordinates $\xi = (\ell - \mu_\alpha) / \sigma_\alpha$. The data collapses perfectly onto the universal theoretical curve $F(\xi) = \exp(-\xi^2)$. The agreement is quantitatively exceptional, with a Pearson correlation coefficient between the measured data and the theory exceeding $\mathbf{99.9\%}$ ($R^2 \approx 1.000$).

The relative error between theory and measurement is less than 1\% for the peak modes and remains below 7\% across all allowed levels, confirming the predictive accuracy of the universal scaling law. In the same plot, the data points for the forbidden levels are shown, falling many orders of magnitude below the theoretical curve, highlighting their profound suppression.

\subsection{Selection Rules and Forbidden Level Suppression}

Figure~\ref{fig:main_validation}(b) provides direct and unambiguous visual evidence for the discrete selection rules. A stark contrast exists between the significant power in the allowed OAM modes (green bars) and the negligible power in the forbidden modes (red bars).

Quantitatively, the suppression of forbidden levels is dramatic. The measured power in these modes is suppressed by a factor exceeding $\mathbf{5 \times 10^4}$ relative to the strongest allowed mode. This robust enforcement demonstrates that only modes with an exact integer match to a bias value $\alpha_k$ can be coherently generated.

Crucially, this suppression is equally effective for the internal forbidden levels ($\ell=-1, 1, 3$) as it is for the external ones. This confirms that the selection rule is a fundamental consequence of orthogonality, not a simple edge effect of the probability distribution.

\subsection{Statistical Validation and Error Analysis}

The integrity of the simulation was confirmed through several statistical checks. The measured probabilities of selecting each bias value $\alpha_k$ deviated from the target Gaussian distribution by less than 0.5\%, confirming accurate sampling.

To test whether the residual power in forbidden levels is purely statistical noise, a chi-squared analysis was performed. The analysis yields a reduced chi-squared value $\chi^2/dof \approx 0.3$. A value of approximately 1 is expected for random statistical fluctuations, so this result confirms that there are no systematic energy leakage effects into the forbidden modes and that their suppression is consistent with the theory.

\subsection{Coherent versus Incoherent Power Analysis}

A distinction between the coherent power $|\langle a_\ell \rangle|^2$ and the total average (incoherent) 
power $\langle |a_\ell|^2 \rangle$ further illuminates the mechanism. While individual realizations scatter 
power broadly across all modes (including forbidden ones) due to chaos, the coherent averaging process cancels 
out the random phase contributions. The final coherence ratio, 
$|\langle a_\ell \rangle|^2 / \langle |a_\ell|^2 \rangle$, is significant for allowed modes but drops to nearly 
zero for forbidden modes. This directly validates that the chaotic component $\delta(r, \theta)$ contributes only 
to an incoherent background that is eliminated by the ensemble average, leaving only the pristine, quantized 
structure from the discrete bias.

\subsection{Universality and Experimental Implications}

The successful validation with a specific, challenging parameter set 
($\mu_\alpha=1.0, \sigma_\alpha=2.0, \sigma=0.5$) gives high confidence in the universal scaling formulation. 
The dimensionless framework allows for the direct and predictive design of discrete chaotic OAM devices across 
different parameter regimes without requiring new, extensive simulations.

The demonstrated robustness of forbidden level suppression ($>10^4$) is of critical importance for experimental 
implementations. It suggests that physical devices can tolerate moderate imperfections in the discretization of 
bias values without compromising the fundamental selection rules, paving the way for practical applications in 
structured light engineering and fundamental tests of rotational symmetry in optical systems.

\section{Conclusion}

In this work, we have developed and rigorously validated a complete theory for the coherent generation of orbital 
angular momentum (OAM) from statistically engineered chaotic optical surfaces. We have demonstrated that, contrary 
to intuition, carefully structured disorder can indeed give rise to a pristine, quantized OAM spectrum.

Our central discovery is a set of fundamental selection rules, born from the orthogonality of the OAM basis, 
which dictate that coherent power can only be channeled into OAM modes that precisely match the integer values 
present in a discrete bias distribution. This principle gives rise to a quantized structure of "allowed" 
and "forbidden" levels. Our theory provides the analytical form for the power spectrum, including a universal 
scaling law, $F(\xi) = \exp(-\xi^2)$, for systems with Gaussian bias statistics.

These theoretical tenets were confirmed with exceptional quantitative accuracy by large-scale Monte Carlo 
simulations. The simulations verified the universal scaling law with a correlation exceeding $99.9\%$ and, 
most critically, demonstrated the robustness of the selection rules. Using a novel test with 
engineered \textit{internal} forbidden levels, we showed that non-allowed modes are suppressed by a factor of 
more than four orders of magnitude, providing definitive proof of the underlying quantization principle.

Ultimately, this work establishes that generating order from chaos in this context is not a matter of chance, 
but of symmetry. The discrete quantization is not an external constraint but an emergent consequence of the 
interplay between the statistical design of the bias and the fundamental rotational symmetry of the system. 
Our validated framework provides a new paradigm for designing disorder-engineered photonic devices and opens 
new avenues for creating table-top optical analogues of phenomena in quantum and gravitational physics.

\section*{Declarations}
All data-related information and coding scripts discussed in the results section are available from the 
corresponding author upon request.

\section*{Disclosures}
The authors declare no conflicts of interest.

\end{document}